\documentclass[11pt,reqno,a4paper]{amsart}
\usepackage{graphicx,epsfig}
\usepackage[foot]{amsaddr}
\usepackage{enumitem}
\usepackage{caption}
\usepackage[numbers,sort&compress]{natbib}
\usepackage{soul}
\usepackage{todonotes}
\usepackage{hyperref}
\usepackage[final,color,notref,notcite]{showkeys}
\usepackage{xparse}

\usepackage[charter]{mathdesign}
\usepackage[scr=rsfso, scrscaled=1]{mathalfa}

\DeclareSymbolFont{Eulerscripteusm10}{U}{eus}{m}{n}
\SetSymbolFont{Eulerscripteusm10}{bold}{U}{eus}{b}{n}
\DeclareMathSymbol{\euW}{\mathord}{Eulerscripteusm10}{"57}
\DeclareMathSymbol{\euD}{\mathord}{Eulerscripteusm10}{"44}
\DeclareMathSymbol{\euM}{\mathord}{Eulerscripteusm10}{"4D}
\DeclareMathSymbol{\euF}{\mathord}{Eulerscripteusm10}{"46}
\DeclareMathSymbol{\euS}{\mathord}{Eulerscripteusm10}{"53}

\DeclareMathAlphabet{\pazocal}{OMS}{zplm}{m}{n}   
\newcommand{\Pcal}{\pazocal{P}}
\newcommand{\Mcal}{\pazocal{M}}
\newcommand{\Dcal}{\pazocal{D}}
\newcommand{\Ncal}{\pazocal{N}}

\newcommand{\Tcal}{\pazocal{T}}
\newcommand{\Acal}{\pazocal{A}}
\newcommand{\Ccal}{\pazocal{C}}

\newtheorem{thm}{Theorem}

\newtheorem{lem}{Lemma}
\newtheorem{defi}{Definition}

\theoremstyle{remark}
\newtheorem{rema}{\bf Remark}
\newtheorem*{cor}{\bf Corollary}
\newtheorem*{nota}{\bf Notation}
\newtheorem{exaa}{\bf Example}

\newenvironment{rem}
  {\pushQED{\qed}\rema}
  {\popQED\endrema}

\newenvironment{exa}
  {\pushQED{\qed}\exaa}
  {\popQED\endexaa}

\makeatletter
\def\@seccntformat#1{%
  \protect\textup{%
    \protect\@secnumfont
    \expandafter\protect\csname format#1\endcsname
    \csname the#1\endcsname
    \protect\@secnumpunct
  }%
}

\makeatother

\definecolor{labelkey}{rgb}{0,.56,.7}

\DeclareDocumentCommand\Cl{ m g }{%
    \IfNoValueF {#2} { \Ccal l_{#1\rightarrow#2}}%
    \IfNoValueT {#2} { \Ccal l_{1\rightarrow#1}}%
    }%

\makeatletter
\ifcase \@ptsize \relax
  \newcommand{\miniscule}{\@setfontsize\miniscule{4}{5}}
\or
  \newcommand{\miniscule}{\@setfontsize\miniscule{5}{6}}
\or
  \newcommand{\miniscule}{\@setfontsize\miniscule{5}{6}}
\fi
\makeatother

\newcommand\SmallMatrix[1]{{\miniscule\arraycolsep=0\arraycolsep\ensuremath{\begin{bmatrix}#1\end{bmatrix}}}}

\newcommand*{\at}{@}

\newcommand{\nn}{\nonumber}
\def\wh{\widehat}
\def\dg{\dagger}
\def\df{\overset{\mathrm{df}}{=}}
\newcommand{\ket}[1]{\mathop{|#1\rangle}\nolimits}
\newcommand{\bra}[1]{\mathop{\left<#1\,\right|}\nolimits}

\newcommand{\kbr}[2]{| #1\rangle\!\langle #2 |}
\newcommand{\diag}{\mathop{{\mathrm{diag}}}}

\newcommand{\Tr}[1]{\mathop{{\mathrm{Tr}}_{#1}}}
\newcommand{\id}{\mathop{{\mathrm{id}}}\nolimits}
\newcommand{\rank}{\mathop{{\mathrm{rank}}}\nolimits}
\newcommand{\im}{\mathop{{\mathrm{im}}}\nolimits}
\newcommand{\dom}{\mathop{{\mathrm{dom}}}\nolimits}
\def\row{\mathop{\mathrm{row}}}
\def\col{\mathop{\mathrm{col}}}

\def\a{\alpha}
\def\b{\beta}

\def\d{\delta}

\def\vr{\varrho}

\def\s{\sigma}

\def\bbC{\mathbb{C}}

\def\bbR{\mathbb{R}}

\def\T{\mathcal{T}}

\def\C{\mathcal{C}}

\newcommand{\msc}[1]{\mathscr{#1}}

\newcommand{\msf}[1]{\mathsf{#1}}

\sodef\so{}{.065em}{.4em plus1em}{2em plus.1em minus.1em}

\begin{document}

\title[]{\so{The pitfalls of deciding whether a quantum  channel is (conjugate) degradable and how to avoid them}}

\begin{abstract}
  To decide whether a quantum channel is degradable is relatively easy: one has to find at least one example of a degrading quantum channel. But in general, no conclusive criterion exists to show the opposite. Using elementary methods we derive a necessary and sufficient condition to decide under what circumstances the conclusion is unambiguous. The findings lead to an extension of the antidegradability region for qubit and qutrit transpose depolarizing channels. In the qubit case we reproduce the known results for the class of qubit depolarizing channels (due to their equivalence). One of the consequences is that the optimal qubit and qutrit asymmetric cloners possess a single-letter quantum capacity formula. We also investigate the ramifications of the criterion for the search of exclusively conjugate degradable channels.
\end{abstract}


\author{Kamil Br\'adler}
\email{kbradler\at uottawa.ca}

\address{Advanced Research Center, University of Ottawa, Ottawa, Canada}
\address{
    Department of Astronomy and Physics,
    Saint Mary's University,
    Halifax, Canada
    }

\maketitle

\thispagestyle{empty}

\section{Introduction}

Quantum noisy channels provide a convenient way of describing open quantum systems. They are at the forefront of interest in quantum information theory~\cite{holevo2012quantum}. A subset of quantum channels called \emph{degradable} channels~\cite{devetakshor2005capacity,cubitt2008structure} was shown to be especially important both from the physical and mathematical point of view. To get an intuitive idea (that will be made precise later in the text) what it means for a channel to be degradable, we recall that a quantum channel $\Mcal$ (as any open quantum system) interacts with its environment. The environment is also an open system and therefore a quantum channel. But this channel shares a lot of features with $\Mcal$~\cite{holevo2012quantum} and it is called a \emph{complementary channel} to $\Mcal$. Then, a channel $\Mcal$ is called degradable if it enjoys a nontrivial property that its complementary channel $\wh\Mcal$ is given by the action of $\Mcal$ itself and another channel $\Dcal$ such that $\wh\Mcal=\Dcal\circ\Mcal$. Hence, $\Mcal$ can be \emph{degraded} to emulate the action of its own environment. As it turns out, many natural physical processes are in fact represented by degradable channels. To name a few, the example is a trivial noiseless channel, the effect of decoherence for a two-level quantum system modeled as a dephasing channel~\cite{devetakshor2005capacity,bennett1997capacities}, the qubit amplitude damping channel describing the information propagation in a spin network~\cite{giovannetti2005information}, optimal universal cloning machines~\cite{bradler2011infinite,bradler2010conjugate}, attenuation and amplification one-mode Gaussian optical channels~\cite{holevo2012quantum} and some fundamental processes from the realm of quantum field theory in  curved spacetime~\cite{bradler2014capacity}.

Apart from its physical prominence, degradable channels play a vital role in the mathematical theory of quantum communication whose central task
is the ultimate rate of reliable quantum communication. This is characterized by the quantum channel capacity~\cite{devetak2005capacity,shor2002quantum,lloyd1997capacity,barnum1998information}. Quantum channel capacity is a fundamental physical quantity that characterizes the ability of a quantum system to coherently transfer a quantum message between a sender and receiver. A great deal of effort has been invested in understanding of its properties~\cite{holevo2012quantum,smith2008quantum}. The problem is that except for degradable channels, the quantum capacity is virtually incalculable.

Interestingly, there is no unambiguous method known to the author to decide whether a degrading channel does \emph{not} exist. More precisely, the non-existence can be unambiguously decided if only if the channel, whose degradability we investigate, satisfies a certain criterion. Here in Sec.~\ref{sec:IFFderivation} we derive the criterion and our approach is based on the representation of quantum channels known as linear superoperators~\cite{havel2003robust,zyczkowski2004duality} introduced in Sec.~\ref{sec:prelim}. If the criterion is not met, we discuss the possibilities of how to proceed in order to disprove degradability but do not provide a conclusive method. That seems to be an interesting open problem. Note that the opposite task of showing degradability is easy even if the criterion is not satisfied: one just needs to find a single instance of a degrading channel and the superoperator formalism is by far the most suitable instrument.

Showing the non-existence of a degrading channel was part of some previous works. Ref.~\cite{smith2007degenerate} uses the superoperator formalism like we do but does not mention whether the calculation  is conclusive (it turns out that it is). Ref.~\cite{cubitt2008structure}, on the other hand, visits this issue more than once and a connection to the non-uniqueness of a degrading channel is emphasized. It is even possible that the criterion derived here is known to the authors  but it is never stated in full clarity as an iff condition (see Sec. II. A of~\cite{cubitt2008structure}).

We illustrate and use the necessary and sufficient condition to extend the parameter range where the complementary channel to the qubit and qutrit transpose depolarizing channel~\cite{datta2006complementarity} is degradable. We show that it contains a subset corresponding to an important class of channels known as the optimal  asymmetric cloners. We are thus able to calculate their quantum capacity in Sec.~\ref{sec:IFFderivation}. Finally, we discuss the implications of our result for the effort of finding exclusively conjugate degradable channels. They form  a different class of channels from degradable channels and their quantum capacity is calculable~\cite{bradler2010conjugate}. It remains to be shown, however, whether it is a mere proper subset of degradable channels~\cite{bradler2010conjugate}. The definition is recalled in Sec.~\ref{sec:prelim} and their link to the properties of bound entangled states derived in~\cite{horodecki2000operational} is further explored.

\section{Preliminaries}\label{sec:prelim}

In this paper we will make an extensive use of two  representations of completely positive maps: the Choi-Jamio\l kowski and superoperator formalism. The two formalisms are essentially identical but it makes sense to distinguish between them as each has advantages the other one lacks. In short, a Choi matrix gives up a quick check whether the map is CP and a linear superoperator is suitable for map composition and inversion by virtue of the standard matrix operations (matrix multiplication and the generalized inverse). The relation between the representations was investigated in~\cite{zyczkowski2004duality} and  we will summarize the most relevant findings. We also point to a few differences in the convention used in this paper, in particular, we will draw the reader's attention to how linear superoperators act on realigned density matrices. For the sake of completeness we will also recall some standard definitions from linear algebra~\cite{horn1991topic}.

\subsection{Some operations on complex matrix spaces}
Let $a\in\bbC^d$ be a $d$-tuple of complex numbers $a=(a_1,\dots,a_d)$. Then $\bbC^d$ is a complex vector space and we assume the choice of the canonical basis. It comes equipped with the inner product
\begin{equation}\label{eq:DotProduct}
  (a,b)\df\sum_{i=1}^d\overline{a}_ib_i,
\end{equation}
where the bar denotes complex conjugation. The inner product has some well documented properties and promotes $\bbC^d$ to a concrete realization of a finite-dimensional abstract Hilbert space. The induced norm $\|a\|_2^2=(a,a)$ is known as the Hilbert-Schmidt or Frobenius norm.  Let $\mathscr{B}(\bbC^d)$ denote the algebra of $d\times d$ complex matrices $\euM_{d}(\bbC)$. $\mathscr{B}(\bbC^d)$ is itself a Hilbert space if equipped with the Hilbert-Schmidt inner product
\begin{equation}\label{eq:HSproduct}
(A,B)\df\Tr{}[A^\dg B],
\end{equation}
where $\dg$ is defined as entry-wise complex conjugation followed by matrix transposition (we will use the symbol $\top$). Again, matrix transposition requires the basis for $A,B\in\euM_{d}(\bbC)$ to chosen and it will be the canonical one $e_{ij}$, with one on the position $(i,j)$ and zero everywhere else. The quantum-mechanical notation $e_{ij}=\kbr{i}{j}$ will be used. The induced matrix norm $\|A\|^2=\Tr{}[A^\dg A]$ is equivalent to the Hilbert-Schmidt norm as can be demonstrated using the following definition.
\begin{defi}\label{def:rowcol}
  Let $\euM_{d_1,d_2}(\bbC)$ be the \emph{set} of $d_1\times d_2$ matrices over $\bbC$ and $A\in\euM_{d_1,d_2}(\bbC)$. We introduce the ``row'' map $\row:\euM_{d_1,d_2}(\bbC)\mapsto\euM_{1,d_1d_2}(\bbC)$ as
  \begin{equation}\label{eq:flatOp}
    \row{[A]}=(A_{11},\dots,A_{1d_2},\dots,A_{d_11},\dots,A_{d_1d_2})
  \end{equation}
  and the ``column'' map $\col:\euM_{1,d_1d_2}(\bbC)\mapsto\euM_{d_1d_2,1}(\bbC)$,  is defined as
    \begin{equation}\label{eq:colOp}
      \col{[A]}=\row{[A]}^{\top}.
    \end{equation}
\end{defi}
Note that $\euM_d\equiv\euM_{d,d}$.
\begin{rem}
  Anticipating the next subsection, we  write
    \begin{subequations}\label{eq:tensorComponentsEx}
    \begin{align}
      A&=A_{k\mu}\kbr{k}{\mu}, \\
      \row{[A]} & = A_{k\mu}\bra{k}\otimes\bra{\mu}\equiv A_{k\mu}\bra{k\mu},\label{eq:tensorComponentsExB}\\
      \col{[A]} & = A_{k\mu}\ket{k}\otimes\ket{\mu}\equiv A_{k\mu}\ket{k\mu},
    \end{align}
    \end{subequations}
    where $1\leq k\leq d_1$ and $1\leq\mu\leq d_2$.
\end{rem}
The  equivalence between $\euM_{1,d^2}(\bbC)$ and $\euM_{d}(\bbC)$ equipped with the inner product Eq.~(\ref{eq:DotProduct}) and~(\ref{eq:HSproduct}), respectively, is then revealed by
\begin{equation}\label{eq:innerProdEquiv}
  \Tr{}[A^\dg B]=\Tr{}[BA^\dg]=\sum_{k,\mu=1}^{d}\overline{A}_{k\mu}B_{k\mu}\equiv(\col{A},\col{B}),
\end{equation}
where $i\equiv k\mu$. The first instance where the above mappings will be used is the following identity~\cite{horn1991topic}.
\begin{lem}\label{lem:columnization}
  Let $A\in\euM_{d_1,d_2}(\bbC),B\in\euM_{d_2,d_3}(\bbC)$ and $C\in\euM_{d_3,d_4}(\bbC)$.
  Then
  \begin{equation}\label{eq:flattenIdentity}
   \col{\big[(ABC)^\top\big]}=(C^\top\otimes A)\col{\big[B^\top\big]}.
  \end{equation}
\end{lem}

One of the key characteristics of any linear map $L:U\mapsto V$ is its \emph{rank} defined as $\rank{L}=\dim{[\im{L}]}$. Rank has many useful properties~\cite{zhang2011matrix} and we summarize some of them. Even more can be said if we think about a matrix $A\in\euM_{d_1,d_2}(\bbC)$ as a concrete realization of a linear map $L$ in the introduced complex vector space, i.e. $A:\bbC^{d_1}\mapsto\bbC^{d_2}$.
\begin{lem}\label{lem:rankProps}
  Let $L:U\mapsto V$ be a linear map and $A,B$ two specific (matrix) linear maps $A:\bbC^{d_1}\mapsto\bbC^{d_2}$ and $B:\bbC^{d_3}\mapsto\bbC^{d_4}$. Then the following properties hold:
  \begin{enumerate}[label=(\roman*)]
    \item $\rank{L}=\dim{U}-\dim{[\ker{L}]}$.\label{lem:rankProps1}
    \item $\rank{A}$ equals the number of nonzero singular values.\label{lem:rankProps2}
    \item $\rank{A}\leq\min{[d_1,d_2]}$.\label{lem:rankProps3}
    \item $\rank{[AB]}\leq\min{[\rank{A},\rank{B}]}$.\label{lem:rankProps4}
    \item $\rank{[A\otimes B]}=\rank{A}\rank{B}$.\label{lem:rankProps5}
    \item $\rank{[A^{-1}]}=\rank{A}$, where $A^{-1}$ denotes the generalized inverse~\cite{ben2003generalized} of $A$.\label{lem:rankProps6}
    \item $\rank{[A^\top]}=\rank{A}$.\label{lem:rankProps7}
    \item Let $C\in\euM_{d_1}\otimes\euM_{d_2}$ and so $C=\sum_{k=1}^nA^{(1)}_k\otimes A^{(2)}_k$. Then the partial transpose over the first subsystem defined (in the canonical basis) as
        $$
        C^{\top_1}=\sum_{k=1}^n\big[A^{(1)}_k\big]^\top\otimes A^{(2)}_k
        $$
        does not preserve rank.\label{lem:rankProps8}
  \end{enumerate}
\end{lem}
\begin{proof}{\ref{lem:rankProps8}}
  An example can easily be found. For $C=\sum_{i,j=0}^1\kbr{ii}{jj}$ we get $\rank{C}=1$ whereas $\rank{[C^{\top_1}]}=4$.
\end{proof}

\subsection{The formalism of finite-dimensional quantum mechanics (briefly)}
Before we put Eq.~(\ref{eq:flattenIdentity}) to use, let's recall some basic building blocks of (mostly finite-dimensional) quantum mechanics~\cite{keyl2002fundamentals,holevo2012quantum}. This will also put the previously introduced matrix operations in a broader context.
Traditionally, the elements of $\mathscr{B}(\bbC^d)$ are called \emph{effects} defined as
\begin{equation}\label{eq:effects}
  \mathscr{F}(\mathscr{B}(\bbC^d))\df\{A\in\mathscr{B};0\leq A\leq\id\},
\end{equation}
where $\id$ is an identity operator. A simple example of an effect is a projector and a generic effect is a POVM element. The set of \emph{states} $\mathscr{S}(\mathscr{B}(\bbC^d))$ is defined as linear functionals $\vr\in\mathscr{S}$ over $\mathscr{B}(\bbC^d)$ by imposing
\begin{align}\label{eq:states}
  \vr(\id)&=1, \nn\\
  \vr(A)&\geq0 \quad (\forall A\geq 0)\in\mathscr{B}(\bbC^d).
\end{align}
The corresponding dual space of states $\mathscr{B}^*$ is paired with $\mathscr{B}$ itself via the inner product introduced earlier in Eq.~(\ref{eq:HSproduct})
\begin{equation}\label{eq:dualIdentification}
  \vr(A)\equiv\Tr{}[\vr A],
\end{equation}
where we will abuse the notation by using the same symbol $\vr$ for an element of $\msc{B}^*$ and $\msc{B}$\footnote{The second requirement of~(\ref{eq:states}) (positive semi-definiteness) singles out a subset of \emph{self-adjoint} elements of $\msc{B^*}$, that is, elements invariant under a star involution. This is because $\msc{B^*}$ is also a $C^*$-algebra. The involution is represented as the $\dg$ operator used earlier in Eq.~(\ref{eq:HSproduct}) and consequently in Eq.~(\ref{eq:dualIdentification}) we assumed~$\vr=\vr^\dg$.}.
Put differently, the bilinear form $\msc{B^*}\times\msc{B}\mapsto\bbC$ (on the left) is equivalent to the Hilbert-Schmidt inner product $\msc{B}\times\msc{B}\mapsto\bbC$ (on the right) by virtue of a complex anti-isomorphism $\msc{B}\mapsto\msc{B}^*$. We proceed in a similar vein for maps between Hilbert spaces and their dual maps. Let $P\in\mathscr{F}(\mathscr{B}(\bbC^{d_B}))$ be an effect and  $M:\mathscr{F}(\mathscr{B}(\bbC^{d_B}))\mapsto\mathscr{F}(\mathscr{B}(\bbC^{d_A}))$ a linear, positive map. Then, the relation
\begin{equation}\label{eq:HeisVSSchr}
\vr(M\circ P)=(M^*\circ\vr)(P)
\end{equation}
defines the dual map $M^*:\mathscr{S}(\mathscr{B}(\bbC^{d_A}))\mapsto\mathscr{S}(\mathscr{B}(\bbC^{d_B}))$. Linearity and the spectral theorem extend the set of effects $\mathscr{F}(\mathscr{B}(\bbC^{d_B}))$ to Hermitian operators (observables) $O\in\mathscr{B}(\bbC^{d_B})$ and from Eqs.~(\ref{eq:HeisVSSchr}) and (\ref{eq:dualIdentification}) we obtain
\begin{equation}\label{eq:dualIdent2}
  \Tr{}[\vr (M\circ O)]=\Tr{}[({M^*\circ\vr}) O],
\end{equation}
where the abuse has been committed again. The map $M^*$ is dual of $M$ but it is actually $M$ that is called the adjoint or \emph{Heisenberg dual} of $M^*$ which itself is a Schr\"odinger evolution operator. Well, almost. It turns out that the RHS of Eq.~(\ref{eq:dualIdent2}) would not be positive for some multipartite states $\vr$ without the restriction to a subset of positive maps called \emph{completely positive (CP) maps} or \emph{quantum channels}. They are defined as positive maps $M^*$ with the additional constraint given by
\begin{equation}\label{eq:CPmap}
   (M^*\otimes\id)\circ\vr_{AR}\geq0
\end{equation}
valid for all $\vr_{AR}\in\mathscr{S}(\mathscr{B}(\bbC^{d_A}\otimes\bbC^{d_R}))$ and $d_A,d_R<\infty$. So it is the set of CP maps that represents physically  sensible evolution operators.

\begin{nota}
In the spirit of Eq.~(\ref{eq:dualIdentification}) we will understand $\mathscr{S}(\mathscr{B}(\bbC^{d_A}))$ as the set of all positive semi-definite matrices of trace one (density matrices) and simplify the notation by writing $\mathscr{S}(A)$  and for multipartite systems $\mathscr{S}(AB\dots Z)$, where, for instance $AB\equiv\msc{B}(\bbC^{d_A})\otimes\msc{B}(\bbC^{d_B})$. To avoid dragging the star  we denote quantum channels in the Schr\"odinger picture by the calligraphic font and so $\Mcal:\mathscr{S}(A)\to\mathscr{S}(B)$ is understood as the action of $M^*$ below (\ref{eq:HeisVSSchr}). We will also reserve the letter $d$ for the Hilbert space dimension, so e.g. $\dim{A}=d_A$.
\end{nota}

To conclude this section we formally introduce a positive map $\C:\mathscr{S}(A)\to\mathscr{S}(A')$ called complex conjugation, where the canonical basis for matrix space $A\simeq A'$ is implicitly present. Complex conjugation is not CP and coincides with the transposition map $\top$ for Hermitian matrices but it is advantageous to introduce distinct notation (cf. the difference between Lemma~\ref{lem:rankProps}, item~\ref{lem:rankProps8}, and Lemma~\ref{lem:ConjDegrSuperOp}).

\subsection{Two (in fact three) representations of a quantum channel}
A remarkable way of representing a quantum channel is known as the Choi-Jamio\l\-kowski isomorphism~\cite{choi1975completely,jamiolkowski1972linear}. Let $\Mcal:\mathscr{S}(A)\to\mathscr{S}(B)$ be the quantum channel. Then there exists a positive semi-definite map $R_\Mcal\in\mathscr{S}(AB)$, sometimes called \emph{Choi matrix},  that represents the action of the channel  via\footnote{We will use the same symbol $\id$ for maps $\id:\euM_d\mapsto\euM_d$ and unnormalized density matrices in which case a subscript denoting the Hilbert space will be attached. Hence ${\id_B\over\Tr{}[\id_B]}\in\msc{S}(B)$.}
\begin{equation}\label{eq:JamiChoi}
  \Mcal\circ\vr_A=\Tr{A}\big[(\vr_A^\top\otimes\id_{B})R_\Mcal\big].
\end{equation}
The channel $\Mcal$ is trace-preserving if its Choi matrix satisfies $\Tr{B}R_\Mcal=\id_A$. Conversely (and this is the trivial direction), any quantum channel $\Mcal:\mathscr{S}(A')\to\mathscr{S}(B)$ ($A'\simeq A$) gives rise to a Choi matrix
\begin{equation}\label{eq:ChoiConverse}
  R_\Mcal=(\id_A\otimes\Mcal)\circ\Phi_{AA'},
\end{equation}
where $\Phi_{AA'}=\sum_{i=1}^{d_A}\ket{i}_A\ket{i}_{A'}$ is an unnormalized maximally entangled state. The big practical advantage of the Choi-Jamio\l kowski formalism is that the verification of complete positivity of a quantum channel reduces to the positive semi-definiteness  of the Choi matrix. At first, the existence of this formalism may seem mysterious. But as investigated in~\cite{zyczkowski2004duality,havel2003robust}, it is closely related to the action of $\Mcal$ represented as a (linear) superoperator. The abundance of details can be found in the cited papers so here we only offer an executive summary. A density matrix $\vr_A\in\msc{S}(A)$ is a  rank-2 tensor and so any CP map  must be representable as a rank-4 tensor. The Choi matrix is indeed a rank-4 tensor and the idea is to realign it to a different rank-4 tensor that will be considered as a linear mapping belonging to $\euM_{d_A^2,d_B^2}$. In this form it will act on a reordered density matrix from $\bbC^{d_A^2}$. This is essentially the superoperator formalism and one of its main advantages is that the action of a quantum channel on a density matrix as well as the composition of two channels are represented by regular matrix multiplication.

There is freedom in the way a linear map can act: from the left or right. Our choice will be to act from the right on the flattened density matrices given by Eq.~(\ref{eq:flatOp}). To this end, let $\Mcal:\msc{S}(A)\to\msc{S}(B)$ be a quantum channel and $R_\Mcal\in\euM_{d_Ad_B}(\bbC)$ its Choi matrix. Since it is a rank-4 tensor we write $R_\Mcal\equiv R_{k\ell;\mu\nu}$ and the channel subscript has been omitted. The interpretation of the indices follows from the convention used in Eq.~(\ref{eq:JamiChoi}) (and correspondingly in~(\ref{eq:ChoiConverse})): the basis of $R_\Mcal$ is ordered as $\kbr{k}{\mu}_A\otimes\kbr{\ell}{\nu}_B$. Then, a superoperator (linear map) $\euM_{d_A^2,d_B^2}\ni\msf{M}:\bbC^{d_A^2}\mapsto\bbC^{d_B^2}$ is obtained by realigning the Choi matrix
\begin{equation}\label{eq:ChoiVSSuperop}
  R_{k\ell;\mu\nu}\to\msf{M}_{k\mu;\ell\nu}.
\end{equation}
We will always use the sans serif font family for the matrix superoperators. The basis of $\msf{M}$ is therefore ordered as
$_A| k\rangle\!\langle \ell |_B\otimes\,_A|\mu\rangle\!\langle \nu |_B\equiv\ket{k\mu}_A\otimes\bra{\ell\nu}_B$. The superoperator in this form is computationally useful since it acts on the transformed input density matrix $\row{[\vr]}$ by $\msf{M}$ from the right (recall that $\row:\euM_{d_A}(\bbC)\mapsto\euM_{1,d_A^2}(\bbC)$):
\begin{equation}\label{eq:SuperOpAction}
  \Mcal\circ\vr\equiv\row{[\vr]}\msf{M}=\vr_{k\mu}\bra{k\mu}\msf{M}_{k\mu;\ell\nu}\ket{k\mu}\otimes\bra{\ell\nu}=\s_{\ell\nu}\bra{\ell\nu}=\row{[\s]}\equiv\s,
\end{equation}
see Eq.~(\ref{eq:tensorComponentsExB}) to decipher the notation. The summation over the repeated \emph{pairs} of indices $k$ and $\mu$ is understood. For two quantum channels $\Mcal$ and $\Ncal$, their composition $\Ncal\circ\Mcal$ is equally easy to obtain as the regular matrix product $\msf{M}\msf{N}$ (be aware of the convention of acting from the right).

The last representation of quantum channels we need to mention is an important result from the times long before quantum information theory existed~\cite{stinespring1955positive}.
\begin{thm}[Stinespring dilation]
    For every completely positive map $\Mcal:\msc{S}(A)\to\msc{S}(B)$ there exists a partial isometry $W_\Mcal:A\to BE$ s.t.
    \begin{equation}
        \Mcal\circ\vr=\Tr{E}[W_\Mcal\vr W_\Mcal^\dg],\quad\forall\vr\in\msc{S}(A).
    \end{equation}
\end{thm}
The channel $\Mcal$ is represented by the action of an isometry $W_\Mcal$ whose target Hilbert space is $BE$. It is sometimes called a \emph{purification} of $\Mcal$. To provide a physical interpretation of the $E$ subsystem, we observe that it is the ``rest of the universe'' (i.e.~environment) that purifies the open quantum dynamics described by $\Mcal$. In principle, the map to the environment is no less significant and this leads to the definition of the complementary channel to $\Mcal$.
\begin{defi}\label{def:Stinespr}
  Let $W_\Mcal:A\to BE$ be a purification of $\Mcal:\msc{S}(A)\to\msc{S}(B)$. Then the complementary channel $\wh{\Mcal}:\msc{S}(A)\to\msc{S}(E)$ is defined as
  \begin{equation}
        \wh{\Mcal}\circ\vr=\Tr{B}[W_\Mcal\vr W_\Mcal^\dg],\quad\forall\vr\in\msc{S}(A).
  \end{equation}
\end{defi}
The environment is unique up to a local isometry  on the $E$ subsystem and so one can find many different forms of $\wh{\Mcal}$. Certain salient features, such as degradability, are preserved, and we will review it in the next subsection. The environment dimension $\dim{E}$, on the other hand, can have any value which bounded from below by  $\rank{[R_\Mcal]}$  also known as the \emph{Choi rank} of $\Mcal$.
\begin{defi}\label{def:unitalCh}
  A quantum channel $\Mcal:\msc{S}(A)\mapsto\msc{S}(B)$ is unital if $\Mcal:{1\over\Tr{}[\id_A]}\id_A\mapsto{1\over\Tr{}[\id_B]}\id_B$.
\end{defi}
The reason for emphasizing the normalization constant will become clear in Example~\ref{exa:TDchannelAgain}. 

\subsection{Degradable and conjugate degradable channels and their relevance}\label{subsec:degrConjdegr}

The main purpose of the machinery from the previous subsections is to find out when it is easy to decide whether a quantum channel is (conjugate) degradable.
\begin{defi}\label{def:allmaps}
Let $\Mcal$ be a quantum channel and $\wh{\Mcal}$ its complementary channel. Then
\begin{enumerate}[label=(\roman*)]
      \item $\Mcal$ is degradable if there exists another channel $\Dcal$ such that $\Dcal\circ\Mcal=\wh\Mcal$. The map $\Dcal$ is called a degrading channel.\label{def:allmaps2}
      \item $\Mcal$ is conjugate degradable if there exists another channel $\overline{\Dcal}$ such that $\overline{\Dcal}\circ\Mcal=\C\circ\wh\Mcal$. The map $\overline{\Dcal}$ is called a conjugate degrading channel.
      \item $\Mcal$ is antidegradable if there exists another channel $\Acal$ such that $\Mcal=\Acal\circ\wh\Mcal$. The map $\Acal$ is called an antidegrading channel.
      \item $\Mcal$ is conjugate antidegradable if there exists another channel $\overline{\Acal}$ such that $\C\circ\Mcal=\overline{\Acal}\circ\wh\Mcal$. The map $\overline{\Acal}$ is called a conjugate antidegrading channel.
    \end{enumerate}
\end{defi}
\begin{figure}[t]
   \resizebox{8cm}{!}{\includegraphics{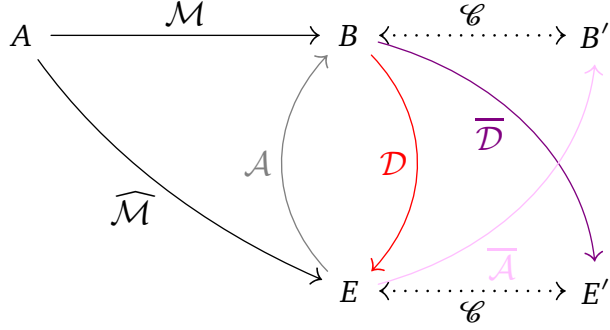}}
    \caption{The action of four maps $\Acal,\Dcal,\overline{\Dcal}$ and $\overline{\Acal}$ relating the output Hilbert spaces of the quantum channel $\Mcal$ and its complementary channel $\wh\Mcal$ is investigated. The channel $\Mcal$ maps the input density matrices from the Hilbert space denoted by $A$ to the output Hilbert space  $B$, whereas its complement's output is another Hilbert space $E$. The involutive map $\C$  acts by transposing (or complex conjugating) the density matrix in the $B$ or $E$ Hilbert space.}
    \label{fig:diagram}
\end{figure}
For a graphical depiction, see Fig.~\ref{fig:diagram}. Note that we can trivially relabel the channels such that $\Ncal=\wh\Mcal$ and so $\wh\Ncal=\Mcal$. Then, if $\Mcal$ is, for example, degradable,  it is equivalent to say that $\Ncal$ is antidegradable. Even though this notational ``permutation'' is trivial, we will often switch the point of view on what is a channel and its complement, especially in Sec.~\ref{sec:clonerCap}. Also note that the class of antidegradable channels has some advantageous properties compared to degradable channels~\cite{cubitt2008structure} and recently an insight into their structure has been gained from the game-theoretic perspective~\cite{buscemi2014game}.

The quantum capacity of a noisy quantum channel $\Mcal$~\cite{divincenzo1998quantum,barnum1998information,devetak2005capacity,lloyd1997capacity} defined as the maximal rate at which quantum information can be sent and perfectly recovered (in the units of bits per channel) is calculated by
\begin{equation}\label{eq:quantCap}
Q(\Mcal)=\lim_{n\to\infty}{1\over n}\max_{\vr}{Q^{(1)}(\Mcal^{\otimes n}(\vr))}
=\lim_{n\to\infty}{1\over n}\max_{\vr}{\big[H(\Mcal^{\otimes n}(\vr))-H(\wh{\Mcal}^{\otimes n}(\vr))\big]},
\end{equation}
where $\vr$ is an input state to $n$ copies of the quantum channel $\Mcal$ and its complement $\wh\Mcal$, the quantity $Q^{(1)}(\Mcal(\vr))$ is called the one-shot quantum capacity also known as the coherent information and $H(\vr)=-\Tr{}[\vr\log{\vr}]$ is the von Neumann entropy.
The magic of degradable channels lies in the observation~\cite{devetakshor2005capacity} that
\begin{equation}\label{eq:singleletterCap}
Q(\Mcal)=\max_{\vr}{Q^{(1)}(\Mcal(\vr))}=\max_\vr{[H(B)_\s-H(E)_\s]},
\end{equation}
where the succinct notation on the right side stresses the fact that the coherent information is maximized over the input ensemble $\vr_A$ but it is evaluated on $\s_{B(E)}$ living in the output Hilbert subspace $B$ and $E$ corresponding to  $\Mcal$ and $\wh\Mcal$, respectively. This is the content behind the statement that the channel capacity is \emph{single-letterized}. The same magic happens for conjugate degradable channels whose quantum capacity is given by Eq.~(\ref{eq:singleletterCap}) as well~\cite{bradler2010conjugate}. The similarity does not end here. If a channel is antidegradable, its quantum capacity is zero. Conjugate antidegradable channels satisfy the same property~\cite{bradler2010conjugate}.

\section{The (non-)uniqueness of degrading and conjugate degrading maps}\label{sec:IFFderivation}

The map composition from item~\ref{def:allmaps2} in Definition~\ref{def:allmaps} rewritten in terms of linear superoperators reads
\begin{equation}\label{eq:degradableSupOp}
  \Dcal\circ\Mcal=\wh\Mcal\leftrightarrow\msf{MD}=\wh{\msf{M}}.
\end{equation}
At first it seems that to decide degradability of $\Mcal$ it suffices to invert the RHS of Eq.~(\ref{eq:degradableSupOp}) and calculate $\msf{D}=\msf{M}^{-1}\wh{\msf{M}}$, where the generalized inverse is used if $d_A\neq d_B$ for $\Mcal:\msc{S}(A)\mapsto\msc{S}(B)$\footnote{We will assume that the inverse satisfies the uniqueness criteria~\cite{ben2003generalized} to avoid further ambiguities.}. Then, $\msf{D}$ is reshuffled to the form of the Choi matrix $R_\Dcal$ and $\Mcal$ is degradable if and only if the eigenvalues of $R_\Dcal$ are non-negative. But this is unfortunately a wrong statement. More precisely, the ``if'' direction is true -- non-negative eigenvalues provide an explicit construction of the degrading map $\Dcal$. But the converse is not correct and it is certainly not an explicit condition that determines whether a channel is
degradable. We will show  under what  circumstances the superoperator formalism can be successfully used to decide the (non-)existence of a degrading map.

The fact that $\msf{D}=\msf{M}^{-1}\wh{\msf{M}}$ does not provide the most general degrading map can be easily seen in the following case. Assume  $\Mcal:\msc{S}(A)\to\msc{S}(B)$ and $\wh{\Mcal}:\msc{S}(A)\to\msc{S}(E)$ such that $d_A<d_B<d_E$. We will call $\msf{D}$ a \emph{candidate} for a degrading map since it must be checked whether it is a CP map and the best way is to investigate the eigenvalues of the Choi matrix. Assume that some of the eigenvalues are negative. Does it exclude the existence of another degrading map? Since $\msf{M}:\bbC^{d_A^2}\mapsto\bbC^{d_B^2}$ and $\wh{\msf{M}}:\bbC^{d_A^2}\mapsto\bbC^{d_E^2}$  we deduce  from Lemma~\ref{lem:rankProps} that
\begin{equation}\label{eq:rankDeficit}
  \rank{[\msf{M}^{-1}\wh{\msf{M}}]}\leq\min{[\rank{\msf{M}},\rank{\wh{\msf{M}}}]}\leq d_A^2,
\end{equation}
where the first inequality follows from items~\ref{lem:rankProps4} and~\ref{lem:rankProps6} and the second one from~\ref{lem:rankProps3}. But this does not prohibit the existence of another degrading candidate, say $\msf{\tilde D}$, where
$$
\rank{\msf{\tilde D}}\leq{d_B^2}
$$
since $\tilde{\msf{D}}:\bbC^{d_B^2}\mapsto\bbC^{d_E^2}$. If $d_A^2<\rank{\msf{\tilde D}}$ it brings the possibility of having another CP degrading map $\tilde\Dcal$ if the  Choi matrix $R_{\tilde\Dcal}$ is positive semi-definite.

But the situation is actually worse. As will be illustrated in Example~\ref{exa:TDchannel} following the formulation of the main result of this section, if some eigenvalues of the Choi matrix  indicate that  $\msf{D}=\msf{M}^{-1}\wh{\msf{M}}$ does not correspond to a CP map, it does not prohibit an existence of a degrading CP map $\tilde\Dcal$ such that $\rank{\msf{D}}=\rank{\msf{\tilde D}}\leq{d_A^2}$.

In the core of the main result lies one of the most elementary results of applied linear algebra addressing the existence of a solution  for a set of linear equations and its uniqueness.
\begin{thm}\label{thm:RoucheCapelli}
  Let $A:\bbR^{d_1}\mapsto\bbR^{d_2}$ acting as $Ax=y$. Assume that $y\in\bbR^{d_2}$ is kept fixed and denote $A'\in\euM_{d_1,d_2+1}$ to be the augmented matrix $[A|y]$ ($y$ added as an additional column of $A$). Then there exists $x\in\bbR^{d_1}$ satisfying $Ax=y$ iff $\rank{A}=\rank{A'}$ and, moreover, $x$ is not unique whenever $d_1>\rank{A}$.
\end{thm}
Two remarks are in order.
\begin{rem}\label{rem:complexComment}
  To generalize the theorem to $A:\bbC^{d_1}\mapsto\bbC^{d_2}$ for $A\in\euM_{d_1,d_2}(\bbC)$ we realize that rank, being the dimension of a subspace, is a field dependent notion. Hence if $\rank{A}=a$  then $\rank{A^\bbR}=2a$ for $A^\bbR\in\euM_{2d_1,2d_2}(\bbR)$ by using $\bbC\simeq\bbR^{\oplus2}$ (the superscript $\bbR$ denotes  ``realification''). So, given $A\in\euM_{d_1,d_2}(\bbC)$ and a fixed $y\in\bbC^{d_2}$, we compare $\rank{A^\bbR}$ and $\rank{[A^\bbR|y^\bbR]}$. If they agree then according to the above theorem there exists $x^\bbR$ satisfying $A^\bbR x^\bbR=y^\bbR$. But that also means that $Ax=y$ is satisfied, where we can explicitly assemble $x$ from $x^\bbR$ once we find it. Note that Theorem~\ref{thm:RoucheCapelli} does not help in any way to find the actual solution. So we can skip the whole procedure and to decide the existence of $x\in\bbC^{d_1}$ we simply check $\rank{A}=\rank{A'}$ for $A\in\euM_{d_1,d_2}(\bbC)$ and $y\in\bbC^{d_2}$.

  In the rest of the paper we assume that $\rank{A}=\rank{A'}$. It is unclear, but unlikely, whether in the context we will employ Theorem~\ref{thm:RoucheCapelli} the rank equality can actually be violated. When it is violated, the system is overdetermined, it is called \emph{inconsistent} and has no solution.
\end{rem}
\begin{rem}
  The second statement is simply a reformulation of~\ref{lem:rankProps1} in Lemma~\ref{lem:rankProps} for $L=A$ and so $d_1-\rank{A}$ is the kernel dimension. With the growing null space the vector $x$ is becoming ``more'' non-unique.
\end{rem}

By applying Eq.~(\ref{eq:colOp}) and considering $\msf{MD}\equiv\msf{MD\id}=\wh{\msf{M}}$, where $\id:\bbC^{d_E^2}\mapsto\bbC^{d_E^2}$, we transform Eq.~(\ref{eq:degradableSupOp}) into
\begin{equation}
  \col{\msf{[MD\msf{\id}]}}=\col{\wh{\msf{M}}}.
\end{equation}
By virtue of Lemma~\ref{lem:columnization} we identify $\msf{M}=C^\top,\msf{D}=B^\top$ and $\id=A^\top$ and find
\begin{equation}\label{eq:columnizationInAction}
  \col{\msf{[MD\msf{\id}]}}=(\msf{M}\otimes\msf{\id})\col{\msf{D}}=\col{\wh{\msf{M}}}.
\end{equation}
Note that the left action of $\msf{M}\otimes\msf{\id}$ on $\col{\msf{D}}$ is not in contradiction with the right action of $\msf{M}$ on $\row{\vr}$ we adopted in~(\ref{eq:SuperOpAction}) -- they obviously act in different contexts.

We are led to the main result.
\begin{thm}\label{thm:IFFcondition}
  Let $\Mcal:\msc{S}(A)\to\msc{S}(B)$ be a quantum channel and $\wh{\Mcal}:\msc{S}(A)\to\msc{S}(E)$ its complementary channel and let the corresponding superoperator $\msf{M}$ of $\Mcal$ be full rank: $\rank{\msf{M}}=\min{[d_A^2,d_B^2]}$. Then, if a degrading map $\Dcal:\msc{S}(B)\to\msc{S}(E)$ exists, it is unique iff $d_B\leq d_A$.
\end{thm}
\begin{rem}
  Restating the obvious, when $d_B\leq d_A$ and the candidate for a degrading map is negative definite (that is, the Choi matrix $R_\Dcal$ obtained by reshuffling  $\msf{D}=\msf{M}^{-1}\wh{\msf{M}}$ has all eigenvalues negative) then no degrading map exists. Conversely, if $d_B>d_A$ then the negative eigenvalues of the Choi matrix are inconclusive for the non-existence of a degrading map.
\end{rem}
\begin{proof}
  We are looking for the condition when $\col{\msf{D}}$ exists and is unique.  According to Theorem~\ref{thm:RoucheCapelli}, the solution exists if
  $$
  \rank{[\msf{M}\otimes\msf{\id}]}=\rank{[\msf{M}\otimes\msf{\id}|\col{\wh{\msf{M}}}]}
  $$
  and so let's assume that. As mentioned in the second paragraph of Remark~\ref{rem:complexComment}, it is not clear whether there are CP maps (meaning physically plausible situations) where this condition can be violated. Following the full rank assumption of $\msf{M}$, we  find from Lemma~\ref{lem:rankProps} (item~\ref{lem:rankProps5}) and the properties of the identity map $\id$ that
    \begin{equation}\label{eq:rankOfTensorProduct}
      \rank{[\msf{M}\otimes\msf{\id}]}=\min{[d_A^2,d_B^2]}\times d_E^2.
    \end{equation}
  Hence $\col{\msf{D}}$ is unique if and only if the RHS of Eq.~(\ref{eq:rankOfTensorProduct}) equals $d_1\equiv\dim{[\col{\msf{D}}]}=d_B^2d_E^2$. But this is equivalent to $d_B\leq d_A$.
\end{proof}
\begin{rem}[Important]  If $\msf{M}$ is rank-deficient ($\rank{\msf{M}}=\d<\min{[d_A^2,d_B^2]}$), then the non-uniqueness of a degrading map appears even for $d_B\leq d_A$. This is because the condition $d_1>\rank{A}$ from Theorem~\ref{thm:RoucheCapelli} is always true: $d_B^2d_E^2>\d d_E^2$. This further adds to the relevance of the issue discussed in this paper -- a plenty of requirements must be met in order to claim that a CP degrading map is unique or it does not exists at all.
\end{rem}
\begin{rem}
  If the result is known to the authors of~\cite{cubitt2008structure} it is stated unfortunately somewhat informally (see Sec.~II.A). In other places the role of the Choi rank (the minimal environment dimension $d_E$~\cite{zyczkowski2004duality}) is often compared to $d_A$ in connection with the uniqueness of the degrading map adding to the impression that it is somehow relevant. Here, the statement of Theorem~\ref{thm:IFFcondition} is unambiguous; the Choi rank of the quantum channel $\Mcal:\msc{S}(A)\to\msc{S}(B)$ has nothing to do with this particular question and there are (and are not) unique degrading candidates for any relation between $d_A$ and $d_E$.

  Of course, by reversing the role of the channel and its complement like will be done in Example~\ref{exa:TDchannel}, the Choi rank becomes relevant again by applying Theorem~\ref{thm:IFFcondition}. This perfectly agrees with another observation in~\cite{cubitt2008structure} (Appendix B.5), where except for two singular cases the antidegradable map for the qubit depolarizing channel is not unique. Its Choi rank is greater than two implying $d_A<d_E$ and according to our result the conclusion follows.
\end{rem}
\begin{lem}\label{lem:ConjDegrSuperOp}
  Let $\msf{D}$ be a candidate for a degrading map $\Dcal$ and the superoperator $\overline{\msf{D}}$ a candidate for a conjugate degrading map $\overline{\Dcal}=\C\circ\Dcal$. Then $\rank{\msf{D}}=\rank{\overline{\msf{D}}}$.
\end{lem}
\begin{proof}
  Density matrix complex conjugation $\C:\mathscr{S}(A')\to\mathscr{S}(A'')$, where $A\simeq A'\simeq A''$, is a specific permutation of the density matrix components. So from Eq.~(\ref{eq:ChoiConverse}) we get
        \begin{align}
          R_\C &= (\id_A\otimes\C)\circ\Phi_{AA'}  =  (\id_A\otimes\C)\circ \sum_{i,j=1}^{d^2_{A}}\ket{ii}_{AA'}\bra{jj}_{AA'}
          =  \sum_{i,j=1}^{d^2_{A}}\kbr{i}{j}_{A}\otimes\kbr{j}{i}_{A''}.
        \end{align}
  But the last equality is just a swap operator $S:\ket{ij}\mapsto\ket{ji}$ and therefore a unitary, and in particular, permutation matrix.  From Eq.~(\ref{eq:ChoiVSSuperop}) we then immediately see that the corresponding (positive but not completely positive) superoperator $\msf{C}$ is the permutation matrix itself
  $$
  \msf{C}=R_\C.
  $$
  This implies  $\overline{\msf{D}}=\msf{D}\msf{C}$ and therefore $\rank{\overline{\msf{D}}}=\rank{\msf{D}}$ since the unitary $\msf{C}$ merely permutes the columns of $\msf{D}$.
\end{proof}
\begin{cor}\label{cor:IFFcondition}
  Let $\Mcal:\msc{S}(A)\to\msc{S}(B)$ be a quantum channel and $\wh{\Mcal}:\msc{S}(A)\to\msc{S}(E)$ its complementary channel. Then, a conjugate  degrading map $\overline{\Dcal}:\msc{S}(B)\to\msc{S}(E')$, satisfying $\overline{\Dcal}\circ\Mcal=\C\circ\wh\Mcal$, exists and it is unique iff $d_B\leq d_A$.
\end{cor}
As a consequence of Lemma~\ref{lem:ConjDegrSuperOp}, the proof of the corollary is identical to that of  Theorem~\ref{thm:IFFcondition} considering the same  assumptions.

Let's take a look at a case where $\msf{D}=\msf{M}^{-1}\wh{\msf{M}}$ fails to produce a CP degrading map but  Theorem~\ref{thm:IFFcondition} admits many solutions. We will indeed find one that is CP. To this end, let's introduce two important quantum channels: qudit transpose depolarizing (TD)~\cite{fannes2004additivity,datta2006complementarity}   and  qudit depolarizing channel
\begin{subequations}
\begin{align}
      \Tcal(\vr)&=t\vr^\top+(1-t){1\over d}\id,\label{eq:transDepChan}\\
      \Pcal(\vr)&=s\vr+(1-s){1\over d}\id,\label{eq:DepChan}
\end{align}
\end{subequations}
where $t,s$ are real parameters and $\id\in\euM_{d}$ is an identity matrix. The complete positivity of $\Tcal$ dictates $-{1\over d-1}\leq t\leq{1\over d+1}$ and that of $\Pcal$ imposes $-{1\over d^2-1}\leq s\leq1$. An important feature of both classes is that they are covariant with respect to the unitary group $SU(d)$. Denote $\msc{G}:SU(d)\to GL(\bbC^d)$ the fundamental representation of $SU(d)$ and $\msc{G}^*:SU^*(d)\to GL(\bbC^d)$ the inequivalent (for $d>2$) fundamental representation. Further denote  $\msc{K}:SU(d)\otimes{SU(d)}\to GL(\bbC^{d^2})$ and  $\msc{L}:SU(d)\otimes{SU^*(d)}\to GL(\bbC^{d^2})$ to be the corresponding tensor product of two fundamental representations. Then, it is known~\cite{datta2006complementarity} that the  qudit TD channel and its complement transform covariantly: $\Tcal\circ\msc{G}  = \msc{G}^*\circ\Tcal$ and $\wh\Tcal\circ\msc{K}  =  \msc{K}\circ\wh\Tcal$. Similarly, $\Pcal\circ\msc{G}  = \msc{G}^*\circ\Pcal$ and $\wh\Pcal\circ\msc{L}  =  \msc{L}\circ\wh\Pcal$ hold. Hence the channels' outputs transform irreducibly but not their complements. Recall that $SU(d)\otimes{SU(d)}$ splits into a direct sum of a completely symmetric and antisymmetric representation and $SU(d)\otimes{SU^*(d)}$ acts irreducibly on $\bbC\oplus\bbC^{d^2-1}$. Their action coincide only for $d=2$ due to the aforementioned lack of an inequivalent fundamental rep in this case. Indeed, by setting $s\mapsto-t$ we find
\begin{equation}\label{eq:DepolSsTD}
\Tcal(\vr)=e^{i\pi/2\s_X}\Pcal(\vr)\big|_{s\mapsto-t}e^{-i\pi/2\s_X}.
\end{equation}

\begin{exa}\label{exa:TDchannel}
    We first find the complementary channel $\wh\Tcal$. It turns out that it is closely related to the following linear \emph{positive} map:
     \begin{equation}\label{eq:mixedsymmetryOp}
         \T:\vr\mapsto\tilde\s={1\over d}\big((\a+\b)\Pi_++(\a-\b)\Pi_-)(\id\otimes\vr)((\a+\b)\Pi_++(\a-\b)\Pi_-\big)^\dg,
     \end{equation}
     where $\Pi_\pm$  is a projector onto $\mathrm{Sym/Alt}[\bbC^d\otimes\bbC^d]$ and $\a,\b\in\bbR$. This map was studied for different purposes in~\cite{vollbrecht2001entanglement}. By restricting to $d=2$ and identifying $\a+\b=\sqrt{1+t\over2}$ and $\a-\b=\sqrt{1-3t\over2}$ we find $t=2\a\b$. The qubit TD channel (and hence its complement) is CP for $-1\leq t\leq1/3$ and only for these values of $t$ the hyperbolae $t=2\a\b$ intersects the normalization ellipse $\a^2+\a\b+\b^2=1$ obtained from~(\ref{eq:mixedsymmetryOp}) by $\Tr{}[\tilde\s]=1$. Hence the normalization condition  automatically ensures the complete positivity of the trace-preserving map $\T$ and so we may write
     \begin{equation}\label{eq:complTDoutput}
    \T\equiv\wh\Tcal(\vr)=
    \left[
        \begin{smallmatrix}
         \frac{1}{2}(1+t)\vr_{11}  & \frac{1+t}{2 \sqrt{2}}\vr_{10}  & 0 & \frac{\sqrt{1-3t} \sqrt{1+t}}{2 \sqrt{2}}\vr_{10}  \\
         \frac{1+t}{2 \sqrt{2}}\vr_{01}  & \frac{1}{4}(1+t) & \frac{1+t}{2 \sqrt{2}}\vr_{10}  & \frac{1}{4}\sqrt{1-3t} \sqrt{1+t}(\vr_{00}-\vr_{11}) \\
         0 & \frac{1+t}{2 \sqrt{2}}\vr_{01}  & \frac{1}{2}(1+t)\vr_{00}  & -\frac{\sqrt{1-3t} \sqrt{1+t}}{2 \sqrt{2}}\vr_{01}  \\
         \frac{ \sqrt{1-3t} \sqrt{1+t}}{2 \sqrt{2}}\vr_{01} & \frac{1}{4}\sqrt{1-3t} \sqrt{1+t}(\vr_{00}-\vr_{11}) & -\frac{\sqrt{1-3t}\sqrt{t+1}}{2\sqrt{2}}\vr_{10} & \frac{1}{4}(1-3t)
        \end{smallmatrix}
    \right].
    \end{equation}
     The channel still deserves to be called a qubit channel in spite of $d_E>2$.

     The question whether $\Tcal$ is degradable can be decided by calculating $\msf{D}=\msf{T^{-1}}\wh{\msf{T}}$ since $d_A=d_B=2$ and Theorem~\ref{thm:IFFcondition} informs us that the eigenvalues of the corresponding Choi matrix provide an unambiguous answer (some are negative except for $t=-1$). But to study antidegradability of $\Tcal$, the question becomes more complicated since the Choi rank of $\Tcal$ equals four.      The candidate for an antidegrading channel $\Acal$ obtained from $\msf{A}=\wh{\msf{T}}^{-1}\msf{T}$   reads
     \begin{align}\label{eq:antidegMapforTD}
      \msf{A}  &=
      \SmallMatrix{
         \frac{3 t^3+t^2+t-1}{2 (t-1) \left(3 t^2+1\right)} & 0 & 0 & \frac{3 t^3-t^2+t+1}{-6 t^3+6 t^2-2 t+2} \\
         0 & \frac{t}{\sqrt{2} (1-t)} & 0 & 0 \\
         0 & 0 & 0 & 0 \\
         0 & \frac{t \sqrt{-3 t^2-2 t+1}}{\sqrt{2} \left(1-t^2\right)} & 0 & 0 \\
         0 & 0 & \frac{t}{\sqrt{2} (1-t)} & 0 \\
         \frac{t+1}{6 t^2+2} & 0 & 0 & \frac{t+1}{6 t^2+2} \\
         0 & \frac{t}{\sqrt{2} (1-t)} & 0 & 0 \\
         \frac{t \sqrt{-3 t^2-2 t+1}}{2-2 t^2} & 0 & 0 & \frac{t \sqrt{-3 t^2-2 t+1}}{2 \left(t^2-1\right)} \\
         0 & 0 & 0 & 0 \\
         0 & 0 & \frac{t}{\sqrt{2} (1-t)} & 0 \\
         \frac{3 t^3-t^2+t+1}{-6 t^3+6 t^2-2 t+2} & 0 & 0 & \frac{3 t^3+t^2+t-1}{2 (t-1) \left(3 t^2+1\right)} \\
         0 & 0 & \frac{t \sqrt{-3 t^2-2 t+1}}{\sqrt{2} \left(t^2-1\right)} & 0 \\
         0 & 0 & \frac{t \sqrt{-3 t^2-2 t+1}}{\sqrt{2} \left(1-t^2\right)} & 0 \\
         \frac{t \sqrt{-3 t^2-2 t+1}}{2-2 t^2} & 0 & 0 & \frac{t \sqrt{-3 t^2-2 t+1}}{2 \left(t^2-1\right)} \\
         0 & \frac{t \sqrt{-3 t^2-2 t+1}}{\sqrt{2} \left(t^2-1\right)} & 0 & 0 \\
         \frac{1-3 t}{6 t^2+2} & 0 & 0 & \frac{1-3 t}{6 t^2+2}}.
     \end{align}
     The three distinct eigenvalues of the corresponding Choi matrix $R_\Acal$ are
     \begin{subequations}\label{eq:antidegrChoiMapEigs}
      \begin{align}
        \lambda^\msf{A}_1=&\frac{-3 t^3+ t^2-t-1}{2 (t-1)(3 t^2+1)},\\
        \lambda^\msf{A}_{2,3}=&\frac{3 t^4+2 t^3+2 t^2+2 t-1\pm2t\sqrt{-18 t^6-6 t^5+t^4-8 t^3-2 t+1}}{2 (t-1) (t+1) \left(3t^2+1\right)}
      \end{align}
     \end{subequations}
     and, for example, for $t=-2/3$ two of them are negative (see Fig.~\ref{fig:qubit} where the eigenvalues are plotted). But it would be incorrect to conclude that $\Tcal$ is not antidegradable. The following superoperator
     \begin{align}\label{eq:Atilde}
      \msf{\tilde A}  &=
        \begin{bmatrix}
         1 & 0 & 0 & 0 \\
         0 & -\frac{1}{\sqrt{2}} & 0 & 0 \\
         0 & 0 & 0 & 0 \\
         0 & -\frac{1}{\sqrt{2}} & 0 & 0 \\
         0 & 0 & -\frac{1}{\sqrt{2}} & 0 \\
         \frac{1}{2} & 0 & 0 & \frac{1}{2} \\
         0 & -\frac{1}{\sqrt{2}} & 0 & 0 \\
         -\frac{1}{2} & 0 & 0 & \frac{1}{2} \\
         0 & 0 & 0 & 0 \\
         0 & 0 & -\frac{1}{\sqrt{2}} & 0 \\
         0 & 0 & 0 & 1 \\
         0 & 0 & \frac{1}{\sqrt{2}} & 0 \\
         0 & 0 & -\frac{1}{\sqrt{2}} & 0 \\
         -\frac{1}{2} & 0 & 0 & \frac{1}{2} \\
         0 & \frac{1}{\sqrt{2}} & 0 & 0 \\
         \frac{1}{2} & 0 & 0 & \frac{1}{2}
        \end{bmatrix}
     \end{align}
     represents a CP map ($\lambda^{\tilde{\msf{A}}}=2$ for $R_{\tilde{\Acal}}$) and satisfies $\wh{\msf{T}}\msf{\tilde A}=\msf{T}$. So it is a legitimate antidegrading map. Also note that $\rank{\msf{A}}=\rank{\msf{\tilde A}}=4$.
\end{exa}
\begin{rem}
    What about the rest of the ``gaps'' of $t$ in Fig.~\ref{fig:qubit}, where the eigenvalues~(\ref{eq:antidegrChoiMapEigs}) of $R_\Acal$ corresponding to $\msf{A}=\wh{\msf{T}}^{-1}\msf{T}$ are negative? Numerical search  yields a CP antidegrading map for all picked values from these intervals. In this case, however, we have just rediscovered the rediscovered. Due to equivalence~(\ref{eq:DepolSsTD}),  $\wh{\Tcal}$ is indeed known~\cite{cerf2000pauli} to be degradable ($\Tcal$ antidegradable) for $t\in[-2/3,1/3]$ as also confirmed in~\cite{cubitt2008structure}.
\end{rem}
\begin{figure}[t]
   \resizebox{8cm}{!}{\includegraphics{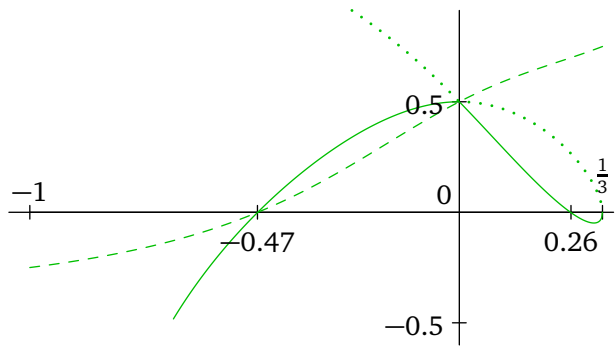}}
    \caption{Three distinctive eigenvalues Eq.~(\ref{eq:antidegrChoiMapEigs}) of the candidate for the antidegrading map $\msf{A}$ in Eq.~(\ref{eq:antidegMapforTD}) of the qubit transpose depolarizing channel are depicted as a function of the depolarizing parameter $t\in[-1,1/3]$.}
    \label{fig:qubit}
\end{figure}
Is there a systematic way of finding out whether a degrading map exists (or no) even if Theorem~\ref{thm:IFFcondition} admits ambiguities? To the author's knowledge, no such procedure is known. The situation is slightly more favorable for $SU(d)$ covariant channels. In this case, the covariance constraint imposed on a candidate for a degrading map leads to the explicit construction of a positive semi-definite matrix that can be tested whether it is a Choi matrix corresponding to a degrading map, see~\cite{bradler2011infinite}. But this is a rather special case.

In principle, the way of exploring all possible solutions is known. We have shown that the ambiguity for degrading and conjugate degrading map comes from the kernel of $\msf{M}\otimes\id$, where $\dim{[\ker{[\msf{M}\otimes\id]}]}=d_E^2(d_B^2-d_A^2)=k$ whenever $d_B>d_A$ and zero otherwise. The basis spanning $\ker{[\msf{M}\otimes\id]}$ can be obtained, for example, from the singular value decomposition of $\msf{M}$ and let's denote the kernel basis  $\{\boldsymbol{b}_i\}_{i=1}^k$. Then  from Eq.~(\ref{eq:columnizationInAction}) we get
\begin{equation}\label{eq:kernelExpandedD}
  (\msf{M}\otimes\msf{\id})\col{\msf{D}}=(\msf{M}\otimes\msf{\id})\big(\col{\msf{D}}+\sum_{i=1}^k\a_i\boldsymbol{b}_i\big)=\col{\wh{\msf{M}}},
\end{equation}
where $\a_i\in\bbC$. This, in turn, leads to the whole family of candidates for a degrading map $\msf{\tilde{D}}$ obtainable by reversing the $\col{}$ operation from Definition~\ref{def:rowcol}. If $R_{\tilde{\Dcal}}$, as a  realigned linear map $\msf{\tilde{D}}$ (reshuffled in the opposite direction of Eq.~(\ref{eq:ChoiVSSuperop})), satisfies the proper requirements to be a valid Choi matrix (positive semi-definite) we declare it to be a valid degrading map. But we can always have quantum channels corresponding to $\msf{M}$ where $\dim{[\ker{[\msf{M}\otimes\id]}]}$ can be arbitrarily large and there is no guarantee that the eigenvalues of the corresponding $\msf{\tilde D}$ can always be calculated for unknown $\a_i$ -- usually on the contrary. A generic numerical search  could only be useful if an actual map is found (see Sec.~\ref{sec:clonerCap}). Otherwise it suffers from the same unambiguity as expressed in Theorem~\ref{thm:IFFcondition}. The obvious exception is if the problem could be reformulated as a semi-definite program. Indeed, the optimizing set is a convex cone but what is missing at the moment is the proper objective function to optimize (if it exists at all).

Before we discuss what our findings imply for the existence of conjugate-degradable quantum channels, let's point out to an intriguing property of certain degrading maps. It could have consequences for the above sketched algorithm to explore all possible candidates for degrading maps.
\begin{exa}\label{exa:TDchannelAgain}
  Let's revisit Example~\ref{exa:TDchannel}. A closer look at the Choi matrix uncovers that $\Tr{}R_\Acal=4/(1+3t^2)$ which seems odd at first (naturally, we consider $t$'s where $\msf{A}$ is positive semi-definite, see Fig.~\ref{fig:qubit}). Recall that in our definition of the Choi matrix (Eq.~(\ref{eq:ChoiConverse})), given $\Acal:\msc{S}(E)\mapsto\msc{S}(B)$ we should have $\Tr{}R_\Acal=\Tr{}\Phi_{EE'}=d_E\equiv4$ and it should not depend on $t$. A different value suggests that $\Acal$ is trace-decreasing and this is also indicated by
    \begin{equation}\label{eq:hiddenTP}
      \Tr{B}R_\Acal={1\over1+3t^2}\diag{[\{1+t,1+t,1+t,1-3t\}]}
    \end{equation}
  (cf. Eq.~(\ref{eq:JamiChoi})). But interestingly, $\Acal$ \emph{is} trace-preserving. To clarify this issue, recall the covariance properties discussed before Example~\ref{exa:TDchannel}. For $d=d_A=2$ we have $\wh\Tcal:\msc{S}(A)\mapsto\msc{S}(E)$ where $\dim{E}=4$. The $SU(2)$ covariance implies that $\im{\wh\Tcal}\varsubsetneq\msc{S}(E)$. Similarly, $\im{\Tcal}\varsubsetneq\msc{S}(B)$ but the case of $\wh\Tcal$ is additionally complicated by  irreducibility of the group action $\msc{K}$ on the target Hilbert space. So the first consequence is that a legitimate antidegrading quantum channel only has to satisfy $\dom{\Acal}=\im{\wh{\Tcal}}$ together with $\im{\Acal}=\im{\Tcal}$. This is perhaps not  that surprising and it has actually nothing to do with~the seeming oddity in~(\ref{eq:hiddenTP}). The real cause is the action of $\wh\Tcal$ on an identity. Due to the reducible action of $\msc{K}$ the channel $\wh\Tcal$ is  not, unlike $\Tcal$, unital (see Definition~\ref{def:unitalCh}) and we get
  \begin{equation}\label{eq:NotAUnitMatrix}
    \wh\Tcal(\id_A)={1\over4}\diag{[\{1+t,1+t,1+t,1-3t\}]}.
  \end{equation}
  Now, since $\Tcal$ is unital, if $\Tr{B}R_\Acal$ in~(\ref{eq:hiddenTP}) was in the form of the unit matrix, as we could naively expect, the condition of unitality on the composite channel $\Tcal=\Acal\circ\wh\Tcal$ would be violated leading to a contradiction.

  We observe that $\Tr{B}R_\Acal=W_\Acal^\dg W_\Acal$, where $W_\Acal:E\mapsto BR$ is an isometric extension of $\Acal$ (see Def.~\ref{def:Stinespr}). We can still call it an isometry, i.e. $\|W_\Acal\nu\|_{BR}=\|\nu\|_E$, where $\|\centerdot\|$ denotes a suitable norm, as long as $\nu$ is restricted to an appropriate subset of $\bbC^{d_E}$.

  Finally, from~(\ref{eq:NotAUnitMatrix}) note that for $t=-1$ and $t=1/3$ everything is ``in order''. This is because $\wh{\Tcal}$ is irreducibly covariant. Similarly for $t=0$ since $\Tcal=\id$ holds. The physical interpretation of $t=0,1/3$ can be found in Section~\ref{sec:clonerCap}.
\end{exa}

\subsection{Quantum capacity of the $d=2,3$ TD channel and the optimal asymmetric cloners}\label{sec:clonerCap}

A subset of the complementary qubit TD channel $\wh\Tcal$ (see below Eq.~(\ref{eq:mixedsymmetryOp})) can be interpreted in an interesting way. We notice that the coefficients $\a$ and $\b$ turn out to be the parameters appearing in the optimal universal asymmetric  $1\to1+1$ cloner for qubits~\cite{niu1998optimal,cerf2000pauli}. There exists a fundamental trade-off for the quality of the clones and the role of $\a,\b$ is to ``tune'' how close in terms of fidelity one of the clones will be  to the input state~\cite{fan2013quantum}. This correspondingly determines the best achievable quality of the other clone. For this purpose it is advantageous to introduce an asymmetry parameter $0\leq p\leq1$~\cite{CerfFiurasek2006} related to $\a$ and $\b$ in the following way:
\begin{subequations}
  \begin{align}
    \a^2 & = {p^2\over{2(1-p+p^2)}}, \\
    \b^2 & = {(1-p)^2\over{2(1-p+p^2)}}.
  \end{align}
\end{subequations}
Hence $t=p(1-p)/(1-p+p^2)$ and we get  the optimal universal symmetric cloner for $p=1/2$~\cite{buvzek1996quantum} corresponding to $t=1/3$. The optimal maximally asymmetric universal qubit cloner is obtained for $p=0,1$ where $t=0$.

Following Example~\ref{exa:TDchannel}, there exists an antidegrading map for $t\in[-2/3,1/3]$. Hence, with the help of the explicit action of $\wh{\Tcal}$, Eq.~(\ref{eq:complTDoutput}), we readily calculate the quantum capacity of $\wh\Tcal$. As revealed in Eq.~(\ref{eq:singleletterCap}) we have to perform the maximization of the coherent information only over a single copy of the channel. But the qubit TD channel $\Tcal$ and its complement $\wh\Tcal$ are also $SU(2)$ covariant and this implies that the maximizing ensemble is a maximally mixed input state of one qubit $\vr=\id_A/2$~\cite{bradler2010conjugate}. Note that we have argued in Example~\ref{exa:TDchannelAgain} that the output of $\wh\Tcal$ does not transform irreducibly. This, however, does not limit the proof presented in~\cite{bradler2010conjugate} and we conclude
\begin{equation}\label{eq:hatTCapacity}
  Q(\wh\Tcal)=H\big(\wh\Tcal(\id/2)\big)-H\big(\Tcal(\id/2)\big)=-3\frac{1+t}{4}\log_2{\frac{1+t}{4}}-\frac{1-3t}{4}\log_2{\frac{1-3t}{4}}-1.
\end{equation}
See Fig.~\ref{fig:CapQubitTDcompl} for the capacity plot.
\begin{figure}[t]
   \resizebox{8cm}{!}{\includegraphics{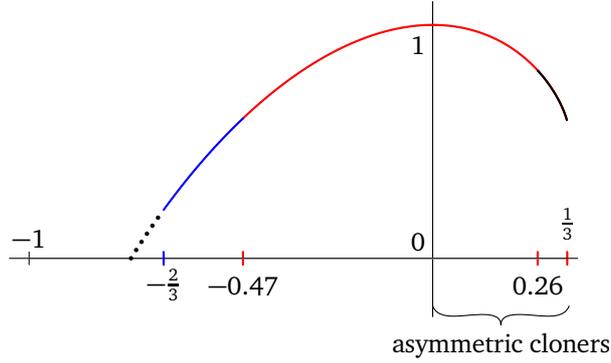}}
    \caption{The quantum capacity $Q(\wh{\Tcal})$ of the complement of the qubit transpose depolarizing channel $\wh\Tcal$ is plotted as a function of the depolarizing parameter $t\in[-1,1/3]$, see Eq.~(\ref{eq:hatTCapacity}). The red (dashed) portion of the curve is where the calculation using the superoperator formalism revealed a degrading map. The blue and black  sections contain numerically found examples of a degrading map and the black section is where conjugate degradability leads to the capacity formula. No degrading channel exists for the dotted portion due to superadditivity of coherent information for the qubit depolarizing channel~\cite{smith2007degenerate} and the relation~Eq.~(\ref{eq:DepolSsTD}).  The section of the parameter space $t\in[0,1/3]$, where $\wh{\Tcal}$ corresponds to the optimal asymmetric qubit cloner.}
    \label{fig:CapQubitTDcompl}
\end{figure}
The one-shot quantum capacity $Q^{(1)}(\wh\Tcal(\vr))$ appearing in Eq.~(\ref{eq:quantCap}) continues to be positive  almost up to $t=-3/4$ but no antidegrading CP map exists beyond $t=-2/3$ (again due to~(\ref{eq:DepolSsTD}) and~\cite{smith2007degenerate}). Finally, notice that for $t=1/3$ the quantum capacity formula reduces to $Q(\wh\Tcal)=\log_2{3}-1$ that was already found in~\cite{bradler2010conjugate} for the optimal (symmetric) qubit cloning channel~$\Cl{2}$ (the red dot in Fig.~\ref{fig:CapQubitTDcompl}).

\begin{figure}[t]
   \resizebox{7cm}{!}{\includegraphics{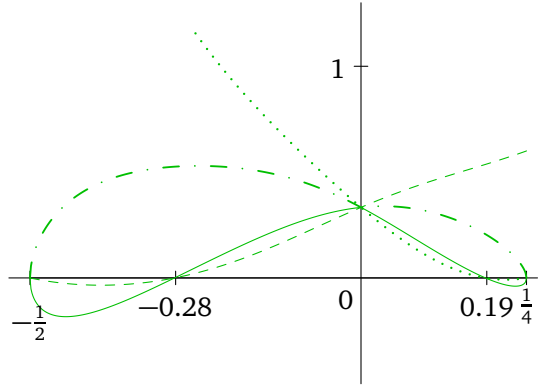}}
    \caption{Three distinctive eigenvalues of the Choi matrix for the antidegrading candidate $\msf{A}$ of the qutrit transpose depolarizing channel are depicted as a function of the depolarizing parameter $t\in[-1/2,1/4]$.}
    \label{fig:qutrit}
\end{figure}

For qutrits ($d=3$), the TD channel is in no way equivalent to the qutrit depolarizing channel. Since $\rank{R_\Tcal}=9$ we have $d_A=3<d_E=9$ and so according to Theorem~\ref{thm:IFFcondition}, a degrading map for $\wh\Tcal$ (antidegrading for $\Tcal$) is not unique. Hence the calculation of the realigned linear operator $\msf{A}=\wh{\msf{T}}^{-1}\msf{T}$ cannot be trusted whenever the eigenvalues are negative. Skipping the details, the eigenvalues are depicted in Fig.~\ref{fig:qutrit} and indeed our distrust is justified. The numerical search reveals a degrading map for any checked value $t\in[-1/2,1/4]$ and we can conclude that with high confidence the qutrit TD channels are all antidegradable. The author is not aware of an analytical proof of this fact but it agrees and extends an earlier observation~\cite{fannes2004additivity} that for $-1/8\leq t\leq1/4$, the TD channel is entanglement-breaking. They are known to be a proper subset of antidegradable channels~\cite{cubitt2008structure}. Antidegradability in the whole interval is further supported by the one-shot quantum capacity calculation being non-negative for $t\in[-1/2,1/4]$ that, based on the numerical evidence, we tentatively declare to be the quantum capacity of~$\wh\Tcal$ (on the whole interval)
\begin{equation}\label{eq:hatTCapacityQutrit}
  Q(\wh\Tcal)
  =-2\frac{1+2t}{3}\log_3{\frac{1+2t}{9}}-\frac{1-4t}{3}\log_3{\frac{1-4t}{9}}-1
\end{equation}
plotted in Fig.~\ref{fig:CapQutritTDcompl}
\begin{figure}[t]
   \resizebox{8cm}{!}{\includegraphics{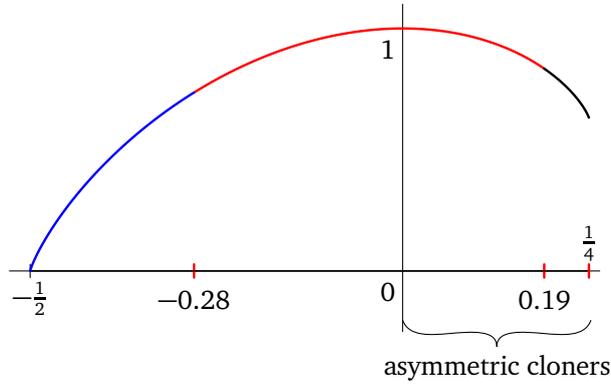}}
    \caption{The quantum capacity $Q(\wh{\Tcal})$ of the complement of the qutrit transpose depolarizing channel is plotted as a function of the (transpose) depolarizing parameter $t\in[-1/2,1/4]$, see Eq.~(\ref{eq:hatTCapacityQutrit}). The red (dashed) portion of the curve is where the calculation using the superoperator formalism revealed a degrading map. The blue  and black  sections contain numerically found examples of a degrading map and the black section is where conjugate degradability leads to the capacity formula. The section of the parameter space $t\in[0,1/4]$, where $\wh{\Tcal}$ corresponds to the optimal asymmetric qutrit cloner.}
    \label{fig:CapQutritTDcompl}
\end{figure}

\subsection{Exclusively conjugate degradable channels}

All conjugate degradable channels known so far are also degradable. It would be highly desirable to find \emph{exclusively} conjugate degradable channel, that is, conjugate degradable channels that are not degradable. Thanks to Theorem~\ref{thm:IFFcondition} we know when the simple calculation of the degrading map candidate $\msf{D}$ in Eq.~(\ref{eq:degradableSupOp}) using the linear superoperator formalism unambiguously reveals whether it is a CP map. Similarly, Corollary on page~\pageref{cor:IFFcondition} informs us about the unambiguous existence of a conjugate degrading quantum channel. As observed in~\cite{bradler2010conjugate}, if $\Mcal$ is conjugate degradable then its complement $\wh\Mcal$ is \emph{PPT} (positive partial transpose).
\begin{defi}
  \begin{enumerate}[label=(\roman*)]\item[] \mbox{}
    \item A channel is PPT if its Choi matrix is PPT.
    \item A bipartite state $\s\in\msc{S}(AB)$ is PPT if $\s^{\top_{A(B)}}\in\msc{S}(AB)$.
  \end{enumerate}
\end{defi}

If $\Mcal$ is an exclusively conjugate degradable channel then its complement is known as \emph{entanglement-binding} channel~\cite{horodecki2000binding}. Its output is a bound entangled state (nonseparable PPT state).  Due to the following result, this immediately gives up the circumstances under which it is hopeless to look for exclusively conjugate degradable channels.
\begin{thm}[\cite{horodecki2000operational}]\label{thm:HoroRank}
  Let $\vr\in\msc{S}(AE)$ such that $\rank{\vr}\leq\max{[d_A,d_E]}$. Then $\vr$ is separable iff it is PPT.
\end{thm}
So the next lemma could be called bad news.
\begin{lem}
  Exclusively conjugate degradable channels can exist only for $d_B>d_A$.
\end{lem}
\begin{proof}
    The complementary channel $\wh\Mcal:\msc{S}(A)\mapsto\msc{S}(E)$ of a hypothetical exclusively conjugate degradable channel $\Mcal:\msc{S}(A)\mapsto\msc{S}(B)$ must necessarily corresponds to a PPT Choi matrix $R_{\wh\Mcal}$. But  according to Theorem~\ref{thm:HoroRank}, only $R_{\wh\Mcal}$ satisfying $\rank{[R_{\wh\Mcal}]}>\max{[d_A,d_E]}$ are not separable and therefore only in this situation it makes sense to try to show that no degrading map exists. But $\rank{[R_{\wh\Mcal}]}$ is also the smallest possible $d_B$. Hence, if $d_A\leq d_E$ then $d_B>d_A$ and Theorem~\ref{thm:IFFcondition} admits non-uniqueness. If $d_A<d_E$ then again $d_B>d_A$ and the kernel containing ambiguities for the conjugate degrading map is even bigger.
\end{proof}
\begin{rem}
  Even stronger result is known for $d_A=d_E=2$. Then, whatever the value of $d_B$ is, there is no bound entangled Choi matrix corresponding to $\wh\Mcal:\msc{S}(A)\mapsto\msc{S}(E)$. This is the celebrated Peres-Horodecki criterion. It is another reason why the qubit TD complement $\wh\Tcal$ cannot be exclusively conjugate degrading in the only remaining possible parameter interval (see the dotted curve in Fig.~\ref{fig:CapQubitTDcompl}).
\end{rem}

\section{Conclusions}

We derived a necessary and sufficient condition to unambiguously decide whether a quantum channel is degradable or conjugate degradable. If the condition is satisfied, the linear superoperator formalism can be very easily used to arrive at the conclusion. In the opposite case, no constructive method seems to be known even though, by again using the superoperator formalism, we showed the roots of the ambiguities and suggested a way to a possible solution. This constitutes an interesting problem for further explorations. The insight obtained in this paper was used to extend the degradability region for the complement to the qubit and qutrit transpose depolarizing channel whose important subset is the optimal asymmetric qubit (qutrit) cloner. Hence we were able to calculate the quantum capacity of all asymmetric qubit and qutrit cloning machines.

The main interesting open problem is the existence of exclusively conjugate degradable channels introduced in~\cite{bradler2010conjugate}. Their quantum capacity is calculable similarly to degradable channels but so far they are not known to form a class of channels on their own. Using the established insights we sharpened the conditions under which they can exist as a separate class and where, on the other hand, would be hopeless to search for them. Their existence is closely related to the properties of bipartite bound entangled states.

\section*{Acknowledgement}
\thanks{The author thanks Vikesh Siddhu and Bob Griffiths for pointing out the importance of the full rank assumption in Theorem~\ref{thm:IFFcondition}.}

\bibliographystyle{unsrt}


\begin{thebibliography}{10}

\bibitem{holevo2012quantum}
A~Holevo.
\newblock {\em Quantum Systems, Channels, Information: A Mathematical
  Introduction}.
\newblock De Gruyter, 2012.

\bibitem{devetakshor2005capacity}
I~Devetak and P~W Shor.
\newblock The capacity of a quantum channel for simultaneous transmission of
  classical and quantum information.
\newblock {\em Communications in Mathematical Physics}, 256(2):287--303, 2005.

\bibitem{cubitt2008structure}
T~S Cubitt, M~B Ruskai, and G~Smith.
\newblock The structure of degradable quantum channels.
\newblock {\em Journal of Mathematical Physics}, 49(10):102104, 2008.

\bibitem{bennett1997capacities}
Ch~H Bennett, D~P DiVincenzo, and J~A Smolin.
\newblock Capacities of quantum erasure channels.
\newblock {\em Physical Review Letters}, 78(16):3217, 1997.

\bibitem{giovannetti2005information}
V~Giovannetti and R~Fazio.
\newblock Information-capacity description of spin-chain correlations.
\newblock {\em Physical Review A}, 71(3):032314, 2005.

\bibitem{bradler2011infinite}
K~Br{\'a}dler.
\newblock An infinite sequence of additive channels: the classical capacity of
  cloning channels.
\newblock {\em IEEE Transactions on Information Theory}, 57(8):5497--5503,
  2011.

\bibitem{bradler2010conjugate}
K~Br{\'a}dler, N~Dutil, P~Hayden, and A~Muhammad.
\newblock Conjugate degradability and the quantum capacity of cloning channels.
\newblock {\em Journal of Mathematical Physics}, 51(7):072201, 2010.

\bibitem{bradler2014capacity}
K~Br{\'a}dler and Ch~Adami.
\newblock The capacity of black holes to transmit quantum information.
\newblock {\em Journal of High Energy Physics}, 2014(5):1--26, 2014.

\bibitem{devetak2005capacity}
I~Devetak.
\newblock The private classical capacity and quantum capacity of a quantum
  channel.
\newblock {\em IEEE Transactions on Information Theory}, 51:44--55, 2005.

\bibitem{shor2002quantum}
P~W Shor.
\newblock The quantum channel capacity and coherent information.
\newblock In {\em Lecture notes, MSRI Workshop on Quantum Computation}, 2002.

\bibitem{lloyd1997capacity}
S~Lloyd.
\newblock Capacity of the noisy quantum channel.
\newblock {\em Physical Review A}, 55(3):1613, 1997.

\bibitem{barnum1998information}
H~Barnum, M~A Nielsen, and B~Schumacher.
\newblock Information transmission through a noisy quantum channel.
\newblock {\em Physical Review A}, 57(6):4153, 1998.

\bibitem{smith2008quantum}
G~Smith, J~Smolin, and A~Winter.
\newblock The quantum capacity with symmetric side channels.
\newblock {\em Information Theory, IEEE Transactions on}, 54(9):4208--4217,
  2008.

\bibitem{havel2003robust}
TF~Havel.
\newblock {Robust procedures for converting among Lindblad, Kraus and matrix
  representations of quantum dynamical semigroups}.
\newblock {\em Journal of Mathematical Physics}, 44(2):534--557, 2003.

\bibitem{zyczkowski2004duality}
K~{\.Z}yczkowski and I~Bengtsson.
\newblock On duality between quantum maps and quantum states.
\newblock {\em Open systems \& information dynamics}, 11(01):3--42, 2004.

\bibitem{smith2007degenerate}
G~Smith and J~A Smolin.
\newblock {Degenerate quantum codes for Pauli channels}.
\newblock {\em Physical Review Letters}, 98(3):030501, 2007.

\bibitem{datta2006complementarity}
N~Datta, M~Fukuda, and A~Holevo.
\newblock Complementarity and additivity for covariant channels.
\newblock {\em Quantum Information Processing}, 5(3):179--207, 2006.

\bibitem{horodecki2000operational}
P~Horodecki, M~Lewenstein, G~Vidal, and I~Cirac.
\newblock Operational criterion and constructive checks for the separability of
  low-rank density matrices.
\newblock {\em Physical Review A}, 62(3):032310, 2000.

\bibitem{horn1991topic}
R~Horn and Ch~Johnson.
\newblock {\em Topics in Matrix Analysis}.
\newblock Cambridge University Press, 1991.

\bibitem{zhang2011matrix}
F~Zhang.
\newblock {\em Matrix theory: basic results and techniques}.
\newblock Springer Science \& Business Media, 2011.

\bibitem{ben2003generalized}
A~Ben-Israel and T~Greville.
\newblock {\em Generalized inverses}, volume~13.
\newblock Springer, 2003.

\bibitem{keyl2002fundamentals}
M~Keyl.
\newblock Fundamentals of quantum information theory.
\newblock {\em Physics Reports}, 369(5):431--548, 2002.

\bibitem{choi1975completely}
M-D Choi.
\newblock Completely positive linear maps on complex matrices.
\newblock {\em Linear Algebra and its Applications}, 10(3):285--290, 1975.

\bibitem{jamiolkowski1972linear}
A~Jamio{\l}kowski.
\newblock Linear transformations which preserve trace and positive
  semidefiniteness of operators.
\newblock {\em Reports on Mathematical Physics}, 3(4):275--278, 1972.

\bibitem{stinespring1955positive}
WF~Stinespring.
\newblock {Positive functions on C*-algebras}.
\newblock {\em Proceedings of the American Mathematical Society},
  6(2):211--216, 1955.

\bibitem{buscemi2014game}
F~Buscemi, N~Datta, and S~Strelchuk.
\newblock Game-theoretic characterization of antidegradable channels.
\newblock {\em Journal of Mathematical Physics}, 55(9):092202, 2014.

\bibitem{divincenzo1998quantum}
D~P DiVincenzo, P~W Shor, and J~A Smolin.
\newblock Quantum-channel capacity of very noisy channels.
\newblock {\em Physical Review A}, 57(2):830, 1998.

\bibitem{fannes2004additivity}
M~Fannes, B~Haegeman, M~Mosonyi, and D~Vanpeteghem.
\newblock Additivity of minimal entropy output for a class of covariant
  channels.
\newblock {\em arXiv preprint quant-ph/0410195}, 2004.

\bibitem{vollbrecht2001entanglement}
K~Vollbrecht and R~Werner.
\newblock Entanglement measures under symmetry.
\newblock {\em Physical Review A}, 64(6):062307, 2001.

\bibitem{cerf2000pauli}
N~J Cerf.
\newblock Pauli cloning of a quantum bit.
\newblock {\em {Physical Review Letters}}, 84(19):4497, 2000.

\bibitem{niu1998optimal}
Ch-S Niu and B~Griffiths.
\newblock Optimal copying of one quantum bit.
\newblock {\em Physical Review A}, 58(6):4377, 1998.

\bibitem{fan2013quantum}
H~Fan, Y-N Wang, L~Jing, J-D Yue, H-D Shi, Y-L Zhang, and L-Z Mu.
\newblock Quantum cloning machines and the applications.
\newblock {\em arXiv preprint arXiv:1301.2956}, 2013.

\bibitem{CerfFiurasek2006}
N~Cerf and J~Fiur{\'a}{\v s}ek.
\newblock Optical quantum cloning.
\newblock In E.~Wolf, editor, {\em Progress in Optics}, volume~49. Elsevier,
  2006.

\bibitem{buvzek1996quantum}
V~Bu{\v{z}}ek and M~Hillery.
\newblock Quantum copying: Beyond the no-cloning theorem.
\newblock {\em Physical Review A}, 54(3):1844, 1996.

\bibitem{horodecki2000binding}
P~Horodecki, M~Horodecki, and R~Horodecki.
\newblock Binding entanglement channels.
\newblock {\em Journal of Modern Optics}, 47(2-3):347--354, 2000.

\end{thebibliography}

\end{document}